\documentclass[conference]{IEEEtran}
\IEEEoverridecommandlockouts

\usepackage[bookmarks=false]{hyperref}
\usepackage{booktabs}
\usepackage{array,tabularx, makecell, multirow}
\usepackage{diagbox,multirow,lineno}
\usepackage{subcaption}
\usepackage{cite}
\usepackage{url}
\usepackage{amsmath,amssymb,amsfonts}
\usepackage{amsthm}
{

	\newtheorem{proposition}{Proposition}

}
\usepackage{algorithmic}
\usepackage{graphicx}
\usepackage{textcomp}
\usepackage{xcolor}
\begin{document}

\title{Version Age-of-Information-based Semantic Secrecy Region Analysis}

\author{\IEEEauthorblockN{Qian Wang\textsuperscript{1}, Bohai Li\textsuperscript{2}, Chao Sun\textsuperscript{1}, Lou Zhao\textsuperscript{1}, Chunshan Liu\textsuperscript{1} and Nikolaos Pappas\textsuperscript{3}
}%

\thanks{The work of Q. Wang was supported in part by the National Natural Science Foundation of China (NSFC) under Grant 62401187. The work of C. Sun was supported in part by the NSFC under Grant 62401186 and the Fundamental Research Funds for the Provincial Universities of Zhejiang under
Grant GK249909299001-028. The work of C. Liu was supported in part by the NSFC under Grant 62471163. The work of N. Pappas was partially supported by ELLIIT.}

\IEEEauthorblockA{\textsuperscript{1}School of Communication Engineering, Hangzhou Dianzi University, Hangzhou, China}

\IEEEauthorblockA{\textsuperscript{2}School of Information Science and Technology, Harbin Institute of Technology (Shenzhen), Shenzhen, China}
\IEEEauthorblockA{\textsuperscript{3}Department of Computer and Information Science, Link\"{o}ping University, Link\"{o}ping, Sweden.}
Emails: \textsuperscript{1}\{qian.wang, sunch, lou.zhao, chunshan.liu\}@hdu.edu.cn, \textsuperscript{2}libohai@hit.edu.cn,
\textsuperscript{3}nikolaos.pappas@liu.se

}

\maketitle

\begin{abstract}
This paper investigates semantic secrecy in wireless status-update systems using Version Age of Information (VAoI). We model the fundamental tradeoff between informative updates and confidentiality (semantic asymmetry between the legitimate destination and the eavesdropper) using a three-node wiretap channel. To evaluate secrecy performance, we introduce the concept of a semantic secrecy region (SSR) and analyze it from a geometrical perspective. By developing a two-dimensional Markov chain, we derive closed-form expressions for the destination's average VAoI and the probability of semantic asymmetry at the eavesdropper. These results enable the joint optimization of transmission probability and power control. For Rayleigh fading channels, we characterize the optimal SSR boundary and power allocation. Numerical results validate our analysis, showing that optimized power allocation significantly enlarges the secrecy region compared with fixed‑power benchmarks.
\end{abstract}


%
\section{Introduction}
The broadcast nature of wireless communication raises severe security and privacy concerns in Internet-of-Things (IoT) networks \cite{sicari2015security}, where an increasing number of devices and physical entities are connected to the Internet. In particular, due to the shared wireless medium, passive eavesdroppers positioned within range can intercept status updates intended for legitimate receivers. Wyner’s wiretap channel model \cite{wyner1975wire} established the information‑theoretic foundation for such scenarios, introducing secrecy capacity and related metrics \cite{bloch2011physical, mohapatra2024physical} that remain fundamental to physical‑layer confidentiality.

Beyond confidentiality, many IoT applications, such as remote monitoring and industrial automation, require the receiver to maintain timely and accurate knowledge of the underlying system state \cite{ageversion}. 
As such, security concerns extend beyond information leakage to freshness and relevance of transmitted status updates \cite{aoie2,aoie1,10620867}. From this perspective, eavesdroppers pose a threat not only by decoding data packets but also by tracking the source process in a timely manner. In other words, even if the system achieves a high secrecy rate, the eavesdropper may still pose a threat if its information is as fresh as the information at the legitimate receiver. 

Freshness-aware secure and covert communication has attracted attention in recent years \cite{aoie2,aoie1,10620867,aoiec1,aoiec2}. Most studies focus on the age of information (AoI), a metric that captures information freshness \cite{Kaul2012real}. For instance, the effect of eavesdropping on status updating was studied in \cite{aoie2}, and the relationship between information freshness and information leakage was investigated in \cite{aoie1}. AoI-based secrecy metrics were further developed in \cite{10620867}. Existing work on covert timely updating mainly focuses on optimizing the age performance at the legitimate receiver under covert or secrecy constraints \cite{aoiec1,aoiec2}. However, AoI measures the time elapsed since the generation of the latest received update and cannot capture the content of the information. To address this limitation, Version Age of Information (VAoI) was introduced in \cite{9517796}. Unlike AoI, VAoI reflects changes in source content, such as new content generation, and measures the version lag, making it more suitable for capturing information semantics. To the best of our knowledge, security-oriented communication designs based on VAoI remain largely unexplored.

In addition, research in physical-layer security has explored confidentiality performance from a spatial perspective \cite{sp1,sp2,sp3}. Existing work uses spatially interpretable measures to characterize secure communication regions, based on secrecy-capacity or secrecy-outage constraints \cite{sp1,sp2}. These studies provide useful insights into how channel conditions and node locations affect the secure transmission range, vulnerability regions \cite{sp2}, and guard zones \cite{sp3}. However, these spatial secrecy formulations are built on rate-based metrics and do not capture the semantic asymmetry between the legitimate receiver and the eavesdropper.

Motivated by this gap, this paper develops a VAoI-based security framework and introduces the concept of semantic secrecy region (SSR) to evaluate secrecy performance from a spatial perspective. Specifically, we consider a three-node wireless status updating system consisting of a source, a legitimate destination, and a passive eavesdropper (Eve). The source generates version updates according to a stochastic process; the destination seeks timely and informative updates, while the eavesdropper attempts to overhear transmissions. The design objective is twofold: the average VAoI at the legitimate destination should remain below a prescribed threshold $\Gamma$, while the probability that the average VAoI at Eve exceeds that at the destination (i.e., the probability of semantic asymmetry) should be larger than $\epsilon$. For given $\Gamma$ and $\epsilon$, the $(\Gamma,\epsilon)$-SSR is defined as the set of Eve's locations for which these two requirements can be satisfied simultaneously. In this sense, the set of coordinates outside the SSR represents spatial vulnerability zones that require protection from the viewpoint of semantic secrecy.

\section{System Model}

We consider a time-slotted wireless status updating system, where a source node $\mathrm{S}$ transmits data packets to a legitimate destination $\mathrm{D}$ in the presence of an eavesdropper $\mathrm{E}$. The packets transmitted by $\mathrm{S}$ come from an information source, and each packet transmission occupies one time slot.
After each transmission, $\mathrm{D}$ returns an ACK/NACK feedback signal to $\mathrm{S}$. The feedback link is assumed to be delay-free and error-free, as commonly assumed in the literature.

Let $h_{\mathrm{SX}}$ denote the channel coefficient from $\mathrm{S}$ to node $\mathrm{X}$, where $\mathrm{X} \in \{\mathrm{D}, \mathrm{E}\}$. The received signal at node $\mathrm{X}$ during time slot $t$ is given by $y_{\mathrm{X}}(t)=h_{\mathrm{SX}}x(t)+n_{\mathrm{X}}(t)$, where $x(t)$ denotes the transmitted signal satisfying $\mathbb{E}\left[|x(t)|^2\right] = P_{\mathrm{S}}$, and $n_{\mathrm{X}}(t)\! \sim\! \mathcal{CN} \left(0, \sigma_{\mathrm{X}}^2\right)$ is the additive white Gaussian noise at node $\mathrm{X}$.
The channel coefficient of the link from $\mathrm{S}$ to node $\mathrm{X}$ is modeled as Rayleigh fading, denoted by $h_{\mathrm{SX}} = \sqrt{\beta_{\mathrm{X}}}g_{\mathrm{X}}$, where $\beta_{\mathrm{X}} = A_{\mathrm{X}}d_{\mathrm{SX}}^{-\eta}$ denotes the large-scale fading coefficient, {$A_{\mathrm{X}}$ is a constant depending on signal and system parameters, $d_{\mathrm{SX}}$ represents the distance between $\mathrm{S}$ and $\mathrm{X}$, and $\eta$ is the path-loss exponent. Moreover, $g_{\mathrm{X}}\!\sim \!\mathcal{CN}(0, 1)$ denotes the small-scale fading coefficient.}
The channels follow a block fading model, which remain constant within each time slot and vary independently across different slots.

$\mathrm{S}$ aims to deliver timely and informative data to $\mathrm{D}$, employing VAoI as the semantic performance metric. VAoI jointly captures the content and the timeliness of information. More specifically, it measures how many versions the received data lags behind $\mathrm{S}$, computed by comparing the latest received version at $\mathrm{D}$ with the current version at $\mathrm{S}$. The arrival of a new version packet is modeled as a Bernoulli process with parameter $\lambda$. If a new version arrives at $\mathrm{S}$, the VAoI increases by $1$. If a version packet is received at $\mathrm{X}$, its VAoI resets to $0$. A new version packet is assumed to arrive at the end of each time slot. We define $I_{arr}(t)=1$ if a new version arrives at slot $t$; $I_{arr}(t)=0$ otherwise. To ensure timely and informative version transmission, only the latest version packet is kept at the source, while the outdated version packets are discarded. A randomized transmission scheme is studied, where $\mathrm{S}$ transmits the version packet with probability $p_{tx}$ if it has a fresher version{; remains inactive, otherwise. }

Meanwhile, $\mathrm{E}$ also attempts to obtain the latest version information by overhearing the wireless transmission. The transmission may suffer from outages due to fading. Let $p$ and $q$ denote the successful reception probabilities of the main channel and the eavesdropping channel, respectively, given by
\begin{equation}
    \begin{aligned}
    p
    \!= \!{\rm Pr}\biggl\{\log\biggl(1+\frac{|h_{\mathrm{SD}}|^2P_{\mathrm{S}}}{\sigma_{\mathrm{D}}^2}\biggr)\ge R_{\mathrm{tar}}\biggr\} =\exp\biggl(-\frac{\alpha_{\mathrm{D}}}{P_{\mathrm{S}}}\biggr),
    \end{aligned}
\end{equation}
\begin{equation}
    \begin{aligned}
    q
    = {\rm Pr}\biggl\{\log\biggl(1+\frac{|h_{\mathrm{SE}}|^2P_{\mathrm{S}}}{\sigma_{\mathrm{E}}^2}\biggr)\ge R_{\mathrm{tar}}\biggr\} =\exp\biggl(-\frac{\alpha_{\mathrm{E}}}{P_{\mathrm{S}}}\biggr),
    \end{aligned}
\end{equation} where $\alpha_{\mathrm{X}}=\left(2^{R_{\mathrm{tar}}}-1\right)d_{\mathrm{SX}}^{\eta}\sigma_{\mathrm{X}}^2/A_{\mathrm{X}}$, $\mathrm{X}\in\{\mathrm{D,E}\}$ and $R_{\mathrm{tar}}$ denotes the targeted spectrum efficiency.

We use $I_{\mathrm{X}}(t)$, ${\mathrm{X}}\in\{\mathrm{D,E}\}$ to denote the successful reception indicators at node $\mathrm{X}$. Specifically, $I_{\mathrm{X}}(t)=1$ if $\mathrm{X}$ successfully receives the transmitted version packet in time slot $t$, and $I_{\mathrm{X}}(t)=0$ otherwise. In this context, the VAoI at the $\mathrm{D}$ and $\mathrm{E}$, denoted by $v_{\mathrm{D}}(t)$ and $v_{\mathrm{E}}(t)$, respectively, can be updated as
\vspace{-0.2cm}
\begin{equation}
v_{\mathrm{X}}(t+1)=
\begin{cases}
0, & \text{if } I_{\mathrm{X}}(t)=1,\ I_{arr}(t)=0, \\
1, & \text{if } I_{\mathrm{X}}(t)=1,\ I_{arr}(t)=1, \\
v_{\mathrm{X}}(t), & \text{if } I_{\mathrm{X}}(t)=0, \ I_{arr}(t)=0, \\
v_{\mathrm{X}}(t)+1, & \text{if } I_{\mathrm{X}}(t)=0, \ I_{arr}(t)=1.
\end{cases}
\end{equation}
where ${\mathrm{X}}\in\{\mathrm{D,E}\}$.

\textit{Definition of Semantic Secrecy Region}: For an average VAoI threshold $\Gamma$ at $\mathrm{D}$ and a semantic freshness asymmetry ratio $\epsilon$, the $(\Gamma,\epsilon)$-SSR is defined as the region in which the probability that the instantaneous VAoI at $\mathrm{E}$ is larger than that at $\mathrm{D}$ is no smaller than $\epsilon$ (i.e., ${\rm Pr}(v_\mathrm{E}>v_\mathrm{D})\geq \epsilon$), given that $\mathbb{E}[v_\mathrm{D}]\leq \Gamma$.

\section{VAoI Analysis}
To characterize the probability ${\rm Pr}(v_\mathrm{E}>v_\mathrm{D})$, we model the system evolution as a two-dimensional Markov chain. The system state in slot $t$ is denoted by $s(t)=(v_\mathrm{D}(t),v_\mathrm{E}(t))$. 
Let $\pi_{i,j}$ denote the steady-state probability of state $(i,j)$, $\forall i,j\geq 0$ under a stationary transmission policy. The state transitions depend on four concurrent events: new version arrivals, source transmission attempts, and the respective reception outcomes at $\mathrm{D}$ and $\mathrm{E}$. To simplify the analysis, we partition the state space into six functional categories: $\mathbf{c}_1=\{(0,0)\}$, $\mathbf{c}_2=\{(1,1)\}$, $\mathbf{c}_3=\{(1,0)\}$, $\mathbf{c}_4=\{(i,j)\}_{i\geq 1,j\geq 1}\setminus \{(1,1)\}$, $
\mathbf{c}_5=\{(i,0)\}_{i>1}$, and $
\mathbf{c}_6=\{(0,j)\}_{j\geq 1}$.

In the fully synchronized state $\mathbf{c}_1=\{(0,0)\}$, the state either remains at $(0,0)$ when no new version arrives or transits to $(1,1)$ when a new version arrives with probability $\lambda$,
\vspace{-0.05cm}
\begin{equation}
\begin{aligned}
P(s'' \mid s=(0,0))=
\begin{cases}
\lambda, &\text{if } s''=(1,1), \\
1-\lambda,& \text{if } s''=(0,0).
\end{cases}
\end{aligned}
\label{eq:c1-transition}
\vspace{-0.05cm}
\end{equation}

In $\mathbf{c}_2=\{(1,1)\}$, i.e., both nodes are exactly one version behind $\mathrm{S}$, the state transition depends on whether a new version arrives, on transmission attempts at the source, and on the reception outcomes at $\mathrm{D}$ and $\mathrm{E}$, given by 
\begin{equation}
\begin{aligned}
&P(s'' \mid s=(1,1)) = \\
&\begin{cases}
(1-\lambda)(1-p_{tx}) + \lambda p_{tx}pq + (1-\lambda)p_{tx}(1-p)(1-q), \\
\multicolumn{1}{r}{\text{if } s''=(1,1),} \\
(1-\lambda) p_{tx}pq, \hfill \quad \quad\text{if } s''=(0,0), \\
(1-\lambda) p_{tx}(1-p)q, \hfill \quad\quad \text{if } s''=(1,0), \\
(1-\lambda) p_{tx}p(1-q), \hfill \quad \quad\text{if } s''=(0,1), \\
\lambda p_{tx}(1-p)q, \hfill \quad\quad  \text{if } s''=(2,1), \\
\lambda p_{tx}p(1-q), \hfill \quad \quad\text{if } s''=(1,2), \\
\lambda p_{tx}(1-p)(1-q) + \lambda (1-p_{tx}), \hfill\quad \quad \text{if } s''=(2,2).
\end{cases}
\end{aligned}
\label{eq:c2-transition}
\end{equation}

In categories $\mathbf{c}_3 \cup\mathbf{c}_5$, $\mathrm{D}$ lags while $\mathrm{E}$ is synchronized. If no new update arrives, a successful reception at the $\mathrm{D}$ resets $v_\mathrm{D}$ to $0$. However, if a new update arrives, $v_\mathrm{E}$ immediately increases to $1$, and $v_\mathrm{D}$ depends on the current transmission. We thus have
\begin{equation}
\begin{aligned}
&P({s'' \mid s=(i,0)\in\mathbf{c}_3\cup\mathbf{c}_5}) =\\
&
\begin{cases}
(1-\lambda)p_{tx}p, & \text{if } s''=(0,0), \\
(1-\lambda)(p_{tx}(1-p) + 1-p_{tx}),
& \text{if } s''=(i,0), \\
\lambda p_{tx}p, & \text{if }  s''=(1,1), \\
\lambda p_{tx}(1\!-\!p)\! +\!\lambda(1\!-\!p_{tx}),
 & \text{if }s''\!\!=(i\!+\!1,1).
\end{cases}
\end{aligned}
\label{eq:c3-transition}
\end{equation}

For states in $\mathbf{c}_4$,  the version lags at $\mathrm{D}$ and $\mathrm{E}$ are behind that at $\mathrm{S}$. State transitions depend on whether receivers can decode the current packet before it is preempted by a new version. A successful reception by either node resets its respective VAoI to either $0$ (if no arrival occurs) or $1$ (if a new version arrives). Therefore, for any $(i,j)\in\mathbf{c}_4$, 
\begin{equation}
\begin{aligned}
&P({s'' \mid s=(i,j)\in\mathbf{c}_4}) =\\
&
\begin{cases}
\lambda p_{tx}pq, & \hspace{-2em} \text{if }  s''=(1,1), \\
\lambda p_{tx}(1-p)q, & \hspace{-2em} \text{if }s''=(i+1,1), \\
(1-\lambda) p_{tx}pq, & \hspace{-2em} \text{if } s''=(0,0), \\
(1-\lambda)( p_{tx}(1-p)(1-q)  +1-p_{tx}), \\
& \hspace{-2em}  \text{if } s''=(i,j), \\
\lambda (p_{tx}(1-p)(1-q)  +1-p_{tx}),
& \hspace{-2em} \text{if }  s''=(i+1,j+1), \\
(1-\lambda) p_{tx}(1-p)q,& \hspace{-2em}  \text{if }  s''=(i,0), \\
(1-\lambda) p_{tx}p(1-q), & \hspace{-2em} \text{if }  s''=(0,j), \\
\lambda p_{tx}p(1\!-\!q), & \hspace{-2em} \text{if } s''=(1,j\!+\!1).
\end{cases}
\end{aligned}
\label{eq:c4-transition}
\end{equation}
As for states in $\mathbf{c}_6$, since $\mathrm{D}$ is fully synchronized, $\mathrm{S}$ will not conduct transmission according to the proposed policy. Thus, the state only changes when a new version arrives, causing $v_\mathrm{D}$ to become $1$ and increasing $v_\mathrm{E}$ by one, where
\begin{equation}
\begin{aligned}
P(\!{s''\! \mid \! s=\!(0,\!j)\!\in\!\mathbf{c}_6}) \!=\!
\begin{cases}
\lambda, \quad \text{if } s''=(1,j+1),\\
1-\lambda, \quad \text{if } s''=(0,j).
\end{cases}
\end{aligned}
\label{eq:c6-transition}
\end{equation}
\subsection{Class-Level Balance Equations}
To derive the steady-state distribution of the Markov chain, we proceed in two stages. First, we determine the aggregate steady-state probability mass for each functional class, denoted by $\pi_{\mathbf{c}_k}$, and then leverage these aggregate results to resolve individual-state probabilities. For the singleton classes, the mappings are $\pi_{\mathbf{c}_1}=\pi_{0,0}$, $\pi_{\mathbf{c}_2}=\pi_{1,1}$, and $\pi_{\mathbf{c}_3}=\pi_{1,0}$.

As the six classes cover the whole state space, the class-level distribution must satisfy the normalization condition
\begin{equation}
\pi_{\mathbf{c}_1}+\pi_{\mathbf{c}_2}+\pi_{\mathbf{c}_3}
+\pi_{\mathbf{c}_4}+\pi_{\mathbf{c}_5}+\pi_{\mathbf{c}_6}=1.
\label{eq:balance-normalization}
\end{equation}
To analyze the distribution, we then group states where $\mathrm{D}$ is fully synchronized ($v_\mathrm{D}=0$), i.e., classes $\mathbf{c}_1$ and $\mathbf{c}_6$. Based on the transition laws, the balance for this subspace is
\begin{equation}
\begin{aligned}
&\pi_{\mathbf{c}_1}+\pi_{\mathbf{c}_6}
= (\pi_{\mathbf{c}_1}+\pi_{\mathbf{c}_6})(1-\lambda) \\
&\qquad+(\pi_{\mathbf{c}_2}+\pi_{\mathbf{c}_3}+\pi_{\mathbf{c}_4}+\pi_{\mathbf{c}_5})
(1-\lambda)p_{tx}p.
\end{aligned}
\label{eq:balance-c1c6}
\end{equation}

Similarly, we then consider the states where $\mathrm{E}$ is synchronized ($v_\mathrm{E}=0$), comprising $\mathbf{c}_3$ and $\mathbf{c}_5$ with the following balance equation
\begin{equation}
\begin{aligned}
&\pi_{\mathbf{c}_3}+\pi_{\mathbf{c}_5}
= (\pi_{\mathbf{c}_2}+\pi_{\mathbf{c}_4})(1-\lambda)p_{tx}(1-p)q \\
& +(\pi_{\mathbf{c}_3}\!+\!\pi_{\mathbf{c}_5})
\left((1\!-\!\lambda)p_{tx}(1\!-\!p)\!+\!(1\!-\!\lambda)(1\!-\!p_{tx})\right).
\end{aligned}
\label{eq:balance-c3c5}
\end{equation}
The individual balances for classes $\mathbf{c}_6$ and $\mathbf{c}_1$ are further refined as follows
\begin{equation}
\pi_{\mathbf{c}_6}
=(\pi_{\mathbf{c}_4}+\pi_{\mathbf{c}_2})(1-\lambda)p_{tx}p(1-q)
+\pi_{\mathbf{c}_6}(1-\lambda),
\label{eq:balance-c6}
\end{equation}
\begin{equation}
\begin{aligned}
\pi_{\mathbf{c}_1}
= {} & (\pi_{\mathbf{c}_3}+\pi_{\mathbf{c}_5})(1-\lambda)p_{tx}p \\
& +(\pi_{\mathbf{c}_4}+\pi_{\mathbf{c}_2})(1-\lambda)p_{tx}pq
+\pi_{\mathbf{c}_1}(1-\lambda).
\end{aligned}
\label{eq:balance-c1}
\end{equation}
For class $\mathbf{c}_2$, the balance accounts for arrivals and the varying success outcomes of transmission attempts.
\begin{equation}
\begin{aligned}
&\pi_{\mathbf{c}_2}\!
=\!\pi_{\mathbf{c}_2}\left((1\!-\!\lambda)(1\!-\!p_{tx}\!+\!p_{tx}(1\!-\!p)(1\!-\!q))
\!+\!\lambda p_{tx}pq\right)\\ &~~~~~~~+\!(\pi_{\mathbf{c}_3}+\pi_{\mathbf{c}_5})\lambda p_{tx}p
+\pi_{\mathbf{c}_1}\lambda+\pi_{\mathbf{c}_4}\lambda p_{tx}pq. \\
\end{aligned}
\label{eq:balance-c2}
\end{equation}

Solving the system of equations \eqref{eq:balance-normalization}--\eqref{eq:balance-c2} yields closed-form expressions for the steady-state probabilities. The mass for the fully synchronized state $(0,0)$ is
\begin{equation}
\begin{aligned}
\pi_{\mathbf{c}_1}=&~\pi_{0,0}= \frac{(1-\lambda)p_{tx}p}{\lambda+(1-\lambda)p_{tx}p} \\
&\times \frac{q(1-(1-p_{tx})(1-\lambda))}
{1-(1-\lambda)(1-p_{tx}+p_{tx}(1-p)(1-q))}.
\end{aligned}
\label{eq:class-c1}
\end{equation}
The probabilities for the lagging states $\pi_{\mathbf{c}_2}$ and $\pi_{\mathbf{c}_3}$ follow similarly, reflecting the impact of the randomized transmission policy and channel conditions.
\begin{equation}
\begin{aligned}
&\pi_{\mathbf{c}_2}=\pi_{1,1}=\frac{(1-(1-\lambda)(1-p_{tx}))p_{tx}pq\lambda}
{\lambda+(1-\lambda)p_{tx}p} \\
& \times \frac{1}
{{(1\!-\!(1\!-\!\lambda)(1\!-\!p_{tx})\!-\!(1\!-\!\lambda)p_{tx}(1\!-\!p)(1\!-\!q))}^2}.
\end{aligned}
\label{eq:class-c2}
\end{equation}
\begin{equation}
\begin{aligned}
\pi_{\mathbf{c}_3}&=\pi_{1,0}
=\frac{(1-\lambda)(1-p)pq p_{tx}^2\lambda}
{{\left(\lambda+(1-\lambda)p_{tx}p\right)}^2} \\
& \times \frac{1}
{{\left(1-(1-\lambda)\left(1-p_{tx}+p_{tx}(1-p)(1-q)\right)\right)}}.
\end{aligned}
\label{eq:class-c3}
\end{equation}
The cumulative mass for the region (i.e., $\mathbf{c}_4\cup\mathbf{c}_2$) where both nodes lag and the individual masses for $\mathbf{c}_5$ and $\mathbf{c}_6$ are determined by substituting these primary results back into the class balance equations, given by
\begin{equation}
\begin{aligned}
&\pi_{\mathbf{c}_4}+\pi_{\mathbf{c}_2}
=   \frac{\lambda}{\lambda+(1-\lambda)p_{tx}p} \\
& \!\times\! \frac{1-(1-\lambda)(1-p_{tx})-(1-\lambda)p_{tx}(1-p)}
{1\!-\!(1\!-\!\lambda)(1\!-\!p_{tx})\!-\!(1\!-\!\lambda)(1\!-\!p_{tx})(1\!-\!p)(1\!-\!q)},
\end{aligned}
\label{eq:c42}
\end{equation}
\begin{equation}
\begin{aligned}
\pi_{\mathbf{c}_5}
\!= \ &
\frac{(1-\lambda)p_{tx}(1-p)q}
{1\!-\!(1\!-\!\lambda)(1\!-\!p_{tx})\!-\!(1\!-\!\lambda)(1\!-\!p_{tx})(1\!-\!p)(1\!-\!q)} \\
&\times\frac{\lambda}{\lambda+(1-\lambda)p_{tx}p} -\pi_{\mathbf{c}_3},
\end{aligned}
\label{eq:class-c5}
\end{equation}
\begin{equation}
\begin{aligned}
\pi_{\mathbf{c}_6}
&= \frac{(1-\lambda)p_{tx}p}{\lambda+(1-\lambda)p_{tx}p} -\frac{(1-\lambda)p_{tx}p}{\lambda+(1-\lambda)p_{tx}p}\\& \quad \times
\frac{q(1-(1-p_{tx})(1-\lambda))}
{1-(1-\lambda)(1-p_{tx}+p_{tx}(1-p)(1-q))}.
\end{aligned}
\label{eq:class-c6}
\end{equation}

\subsection{Representative Single-State Probabilities}
With the aggregate class masses available, we next resolve representative individual-state probabilities. This derivation begins with the entry state of class $\mathbf{c}_6$, proceeds to the boundary state $(1,2)$ within the lagging region, and extends to a general recursive form for all states $(1,j)$ where $j \geq 2$.

We first analyze $\pi_{0,1}$, as it serves as the foundation for the subsequent recursion. The inflow into state $(0,1)$ originates from synchronized successes in class $\mathbf{c}_2$ and the first positive row of $\mathbf{c}_4$ (i.e., states $(i,1)$ with $i\geq 2$), we have
\vspace{-0.05cm}
\begin{equation}
\pi_{0,1}
=\left(\pi_{1,1}+\sum_{i=2}^{\infty}\pi_{i,1}\right)(1-\lambda)p_{tx}p(1-q)
+\pi_{0,1}(1-\lambda),
\label{eq:pi01-balance}
\vspace{-0.1cm}
\end{equation}
where $\pi_{1,1}=\pi_{\mathbf{c}_2}$ is already known from \eqref{eq:class-c2}. The remaining term is the total probability mass on states $(i,1)$ with $i\geq 2$. By collecting the inflow from $\mathbf{c}_4$, the axis $\mathbf{c}_3\cup\mathbf{c}_5$, and the self-loop within the first positive row, we obtain
\vspace{-0.05cm}
\begin{equation}
\begin{aligned}
&\sum_{i=2}^{\infty}\pi_{i,1}
\!=\!\sum_{i=2}^{\infty}\pi_{i,1}(1-\lambda)
\left(1-p_{tx}+p_{tx}(1-p)(1-q)\right)+\\ &(\pi_{\mathbf{c}_4}\!+\!\pi_{\mathbf{c}_2})\lambda p_{tx}(1\!-\!p)q \! +\!(\pi_{\mathbf{c}_3}\!+\!\pi_{\mathbf{c}_5})\lambda(1\!-\!p_{tx}\!+\!p_{tx}(1\!-\!p)).
\end{aligned}
\label{eq:sum-i1}
\end{equation}
Substituting \eqref{eq:class-c5} and \eqref{eq:c42} into \eqref{eq:sum-i1}, and then substituting the result together with \eqref{eq:class-c2} into \eqref{eq:pi01-balance}, yields
\begin{equation}
\pi_{0,1}
=\frac{(1-\lambda)(1-q)q p_{tx}^2p}
{{\left(1-(1-\lambda)(1-p_{tx}+p_{tx}(1-p)(1-q))\right)}^2}.
\label{eq:pi01}
\end{equation}

We next derive $\pi_{1,2}$,  which represents the primary boundary in the region where $\mathrm{E}$ lags further than $\mathrm{D}$ in $\mathbf{c}_4$ ($j \geq i$). Its balance equation incorporates inflows from $(0,1)$ and the entire first positive row, given by
\begin{equation}
\begin{aligned}
\pi_{1,2}
= \ & \pi_{1,2}(1-\lambda)
\left(1-p_{tx}+p_{tx}(1-p)(1-q)\right) \\
&+\pi_{0,1}\lambda+\sum_{i=1}^{\infty}\pi_{i,1}\lambda p_{tx}p(1-q).
\end{aligned}
\label{eq:pi12-balance}
\end{equation}
Using \eqref{eq:pi01}, \eqref{eq:class-c2}, and \eqref{eq:sum-i1}, we obtain
\begin{equation}
\pi_{1,2}
=\frac{(1-q)q p \lambda p_{tx}^2}
{{\left(1-(1-\lambda)(1-p_{tx}+p_{tx}(1-p)(1-q))\right)}^3}.
\label{eq:pi12}
\end{equation}

We then proceed to derive $\pi_{0,2}$. For the states $(i,2)$ with $i\geq 2$, namely the second positive row excluding the boundary state, the total probability mass satisfies
\begin{equation}
\begin{aligned}
& \sum_{i=2}^{\infty}\pi_{i,2}
=  \sum_{i=1}^{\infty}\pi_{i,1}\lambda
\left(1-p_{tx}+p_{tx}(1-p)(1-q)\right) \\
& +\sum_{i=2}^{\infty}\pi_{i,2}(1-\lambda)
\left(1-p_{tx}+p_{tx}(1-p)(1-q)\right).
\end{aligned}
\label{eq:sum-i2}
\end{equation}
The balance equation for $(0,2)$ is
\vspace{-0.05cm}
\begin{equation}
\pi_{0,2}=\sum_{i=1}^{\infty}\pi_{i,2}(1-\lambda)p_{tx}p(1-q)+\pi_{0,2}(1-\lambda).
\label{eq:pi02}
\end{equation}
Substituting \eqref{eq:sum-i2} into \eqref{eq:pi02} shows that $\pi_{0,2}$ is obtained from $\pi_{0,1}$ through the factor
\begin{equation}
\gamma
=1-\frac{p_{tx}q}
{{\left(1-(1-\lambda)(1-p_{tx}+p_{tx}(1-p)(1-q))\right)}}.
\label{eq:gamma}
\end{equation}
That is, $\pi_{0,2}=\pi_{0,1}\gamma$.
The next state in the sequence is $\pi_{1,3}$. Its balance equation is
\begin{equation}
\begin{aligned}
\pi_{1,3}
=  & \pi_{1,3}(1-\lambda)
\left(1-p_{tx}+p_{tx}(1-p)(1-q)\right) \\
& +\pi_{0,2}\lambda+\sum_{i=1}^{\infty}\pi_{i,2}\lambda p_{tx}p(1-q).
\end{aligned}
\label{eq:pi13-balance}
\end{equation}
Substituting $\pi_{0,2}$ and \eqref{eq:sum-i2} into \eqref{eq:pi13-balance} yields $\pi_{1,3}=\pi_{1,2}\gamma$.
By induction, we have
\begin{equation}
\pi_{0,j}=\pi_{0,1}\gamma^{j-1},\ j\geq 1,\text{and } \pi_{1,j}=\pi_{1,2}\gamma^{j-2},\ j\geq 2.
\label{eq:pi0j}
\end{equation}

We next turn to the diagonal states in $\mathbf{c}_4$. For $i\geq 1$, the state $(i+1,i+1)$ can only be reached from $(i,i)$ when a new version arrives, and the transmission does not take place or fails at both receivers; the same event also gives the self-loop at $(i+1,i+1)$ when no new version arrives. Therefore,
\begin{equation}
\begin{aligned}
&\pi_{i+1,i+1}
 = \pi_{i,i}\lambda\left(1-p_{tx}+p_{tx}(1-p)(1-q)\right) \\
& +\pi_{i+1,i+1}(1-\lambda)\left(1-p_{tx}+p_{tx}(1-p)(1-q)\right).
\end{aligned}
\label{eq:piii-balance}
\end{equation}
Rearranging \eqref{eq:piii-balance} yields
$\pi_{i,i}=\pi_{1,1}m^{i-1}$, $i\geq 1$,
where 
$m=\frac{\lambda\left(1-p_{tx}+p_{tx}(1-p)(1-q)\right)}
{1-(1-\lambda)\left(1-p_{tx}+p_{tx}(1-p)(1-q)\right)}$.

\vspace{0.07cm}
We finally extend the recursion from the boundary of $\mathbf{c}_4$ to its interior. Consider a state $(i,j)\in\mathbf{c}_4$ with $j>i$, which lies in the region $j\geq i$. From \eqref{eq:c4-transition}, the only predecessor that reaches $(i,j)$ through a new version arrival is $(i-1,j-1)$, while the self-loop at $(i,j)$ occurs when no new version arrives and the transmission either does not take place or fails at both receivers. Therefore, the local balance equation is
\begin{equation}
\begin{aligned}
\pi_{i,j}
= {} & \pi_{i-1,j-1}\lambda\left(1-p_{tx}+p_{tx}(1-p)(1-q)\right) \\
& +\pi_{i,j}(1-\lambda)\left(1-p_{tx}+p_{tx}(1-p)(1-q)\right).
\end{aligned}
\label{eq:piij-balance}
\end{equation}
Rearranging \eqref{eq:piij-balance} gives the diagonal recursion $\pi_{i,j}=\pi_{i-1,j-1}m$ and iterating the recursion along the diagonal until $i=1$ yields
\begin{equation}
\pi_{i,j}=\pi_{1,j-i+1}{m}^{i-1},\qquad j>i.
\label{eq:piij-j-greater-i}
\end{equation} 

So far, we have characterized the steady-state distribution for the states $j\geq i$. These results provide the necessary components to compute the probability of semantic asymmetry, i.e., $\mathrm{Pr}(v_{\mathrm{E}} > v_{\mathrm{D}})$. The derivation of steady-state distribution for states $\{(i,j)\}_{i>j}$ is similar to the derivation for states $\{(i,j)\}_{i<j}$ and is omitted due to the space limit. 

\subsection{$(\Gamma,\epsilon)$-SSR Analysis}
To determine the $(\Gamma,\epsilon)$-SSR, we first analyze the probability of semantic asymmetry $\mathrm{Pr}(v_{\mathrm{E}} > v_{\mathrm{D}})$. Based on the steady-state distribution derived above, we have
\begin{equation}\label{prt}
\begin{aligned}
\mathrm{Pr}(v_{\mathrm{E}} \!> \!v_{\mathrm{D}})
\!\overset{a}{=}  \! 
\sum_{i=0}^{\infty}\!\left(\!\pi_{0,1}\gamma^{i}+\!\sum_{j=2}^{\infty}\pi_{1,j}m^{i}\!\!\right) 
\!\!\overset{b}{=}\!\!\frac{p\!-\!pq}{p\!+\!q\!-\!pq},
\end{aligned}
\end{equation} where $\overset{a}{=}$ is obtained by substituting \eqref{eq:pi0j} and \eqref{eq:piij-j-greater-i} into \eqref{prt}, and $\overset{b}{=}$ is obtained by substituting \eqref{eq:pi01}, \eqref{eq:pi12} and \eqref{eq:pi0j} together with \cite[Eq. 0.112]{gradshteyn2014table}.  By observation, \eqref{prt} is consistent with the system behavior. In particular, when $q=0$, $\mathrm{E}$ cannot overhear any transmission and $\mathrm{Pr}(v_{\mathrm{E}} > v_{\mathrm{D}})=1$. In contrast, when $q=1$, $\mathrm{E}$ perfectly overhears the channel and $\mathrm{Pr}(v_{\mathrm{E}} > v_{\mathrm{D}})=0$. These two boundary cases agree with intuition and support the correctness of the analysis.

After that, we analyze $\mathbb{E}[v_{\mathrm{D}}]$. Since the VAoI evolution at $\mathrm{D}$ does not depend on the eavesdropper's reception process, the analysis reduces to a one-dimensional Markov chain. Define $\pi_{i,\cdot} \triangleq \sum_{j} \pi_{i,j}$ as the steady-state probability of $v_{\mathrm{D}}=i$. From \eqref{eq:class-c1}-\eqref{eq:class-c3}, \eqref{eq:class-c6}, \eqref{eq:pi12} and \eqref{eq:pi0j}, we first obtain $\pi_{0,\cdot}=\frac{(1-\lambda)p_{tx}p}{\lambda+(1-\lambda)p_{tx}p}$ and $\pi_{1,\cdot}=\frac{\lambda p_{tx}p}{\left(\lambda+(1-\lambda)p_{tx}p\right)^2}$. For the transition from state $(i,\cdot)$ to $(i+1,\cdot)$, for all $i\geq 1$, the balance relation is $\pi_{i+1,\cdot}=\left(\pi_{i,\cdot}\lambda+\pi_{i+1,\cdot}(1-\lambda)\right)(1-p_{tx}+p_{tx}(1-p))$. Applying this relation recursively yields
\vspace{-0.07cm}
\begin{equation}
    \label{eq:idots}
    \pi_{i+1,\cdot}=\pi_{i,\cdot}\left(\frac{\lambda(1-p_{tx}p)}{\lambda+(1-\lambda)p_{tx}p}\right)^i.
\vspace{-0.05cm}
\end{equation} 
Then, $\mathbb{E}[v_{\mathrm{D}}]=\sum_{i=0}^\infty i \pi_{i,\cdot} $ can be calculated as follows,
\vspace{-0.03cm}
\begin{equation}
    \label{eq:vd}
    \mathbb{E}[v_{\mathrm{D}}]\!=\!\sum_{i=1}^\infty i \pi_{1,\cdot}\!\left(\!\frac{\lambda(1\!-\! p_{tx}p)}{\lambda\!+\!(1\!-\!\lambda)p_{tx}p}\!\right)^{i\!-\!1}\!\overset{c}{=}\!\frac{\lambda}{p_{tx}p},
\vspace{-0.05cm}
\end{equation} where $\overset{c}{=}$ is obtained by adopting \cite[Eq.0113]{gradshteyn2014table}. \eqref{prt} and \eqref{eq:vd} together provide the key quantities required to characterize the $(\Gamma,\epsilon)$-SSR.

To study the $(\Gamma,\epsilon)$-SSR under joint design of the transmission probability and the transmit power, we write the successful reception probabilities as $p(P_{\mathrm{S}})$ and $q(P_{\mathrm{S}},\mathbf{w}_{\mathrm{E}})$, where $p_{tx}\in(0,1]$, $P_{\mathrm{S}}\in(0,P_{\max}]$ and $\mathbf{w}_{\mathrm{E}}$ denotes the location of $\mathrm{E}$. 
Accordingly, the $(\Gamma,\epsilon)$-SSR is defined as 
\begin{equation}
\begin{aligned}
&\mathcal{S}(\Gamma,\epsilon)
\!=\!\bigl\{\mathbf{w}_{\mathrm{E}}| \exists (p_{tx},P_{\mathrm{S}}), p_{tx}\!\in\!(0,1],\ P_{\mathrm{S}}\!\in\!(0,P_{\max}], \\
& ~~\text{s.t.} ~~
\frac{\lambda}{p_{tx}p(P_{\mathrm{S}})}\le \Gamma, \\& ~~~~~~~~\frac{p(P_{\mathrm{S}})-p(P_{\mathrm{S}})q(P_{\mathrm{S}},\mathbf{w}_{\mathrm{E}})}
{p(P_{\mathrm{S}})+q(P_{\mathrm{S}},\mathbf{w}_{\mathrm{E}})-p(P_{\mathrm{S}})q(P_{\mathrm{S}},\mathbf{w}_{\mathrm{E}})} \ge \epsilon \bigr\}.
\end{aligned}
\label{eq:ssr-def}
\end{equation} 
We first examine the role of the transmission probability $p_{tx}$. From the destination semantic constraint, we have 
\begin{equation}
\frac{\lambda}{p_{tx}p(P_{\mathrm{S}})}\le \Gamma
\quad \Longleftrightarrow \quad
p(P_{\mathrm{S}})\ge \frac{\lambda}{\Gamma p_{tx}}.
\label{eq:ps-feasible}
\end{equation}
Since $p(P_{\mathrm{S}})$ increases with $P_{\mathrm{S}}$, \eqref{eq:ps-feasible} induces a lower bound on the feasible transmit power. Let $P_{\min}(p_{tx})$ denote the minimum power satisfying the destination freshness constraint 
$p\!\left(P_{\min}(p_{tx})\right)=\frac{\lambda}{\Gamma p_{tx}}$. Thus, for a given $p_{tx}$, the feasible power interval is $[P_{\min}(p_{tx}),P_{\max}]$, and a larger $p_{tx}$ leads to a smaller $P_{\min}(p_{tx})$.

This observation separates the joint design into two steps. For each fixed $p_{tx}$, the SSR is determined by optimizing $P_{\mathrm{S}}$ over the feasible interval $[P_{\min}(p_{tx}),P_{\max}]$. Then, the optimal value of $p_{tx}$ is selected by comparing the resulting feasible sets. Since $\mathrm{Pr}(v_{\mathrm{E}}>v_{\mathrm{D}})$ depends on $P_{\mathrm{S}}$ but not on $p_{tx}$, the role of $p_{tx}$ is only to determine the admissible range of $P_{\mathrm{S}}$.
\begin{proposition}\label{prop:ssr-reduction}
Assume that $p_{tx}\in(0,1]$ and $P_{\mathrm{S}}\in(0,P_{\max}]$. Then the following statements hold.
\begin{enumerate}
\item The destination freshness constraint $\mathbb{E}[v_{\mathrm{D}}]\le \Gamma$ is equivalent to $P_{\mathrm{S}}\ge P_{\min}(p_{tx})$.
\item Since $P_{\min}(p_{tx})$ decreases with $p_{tx}$, the smallest feasible transmit power is attained at $p_{tx}=1$.
\item A feasible design exists if and only if
$P_{\max}\ge P_{\min}(1)$, or, equivalently,
$p(P_{\max})\ge \frac{\lambda}{\Gamma}$.
\item For any fixed $p_{tx}$, the optimal transmit power is obtained from $P_{\mathrm{S}}^\star(p_{tx})
=
\arg\max_{P_{\mathrm{S}}\in[P_{\min}(p_{tx}),P_{\max}]}
|\mathcal{S}(\Gamma,\epsilon;P_{\mathrm{S}})|$.

\item Since $P_{\min}(p_{tx})$ decreases with $p_{tx}$ and $\mathrm{Pr}(v_{\mathrm{E}}>v_{\mathrm{D}})$ does not depend directly on $p_{tx}$, $p_{tx}^\star=1$ is one of the optimal transmission probabilities.
\end{enumerate}
\end{proposition}
Proposition~\ref{prop:ssr-reduction} clarifies the roles of $p_{tx}$ and $P_\mathrm{S}$. The transmission probability does not appear directly in the secrecy metric $\mathrm{Pr}(v_{\mathrm{E}}>v_{\mathrm{D}})$. Instead, it affects the SSR indirectly by changing the feasible range of $P_\mathrm{S}$ through $P_{\min}(p_{tx})$. Therefore, the joint design reduces to a power optimization over the interval $[P_{\min}(1),P_{\max}]$.

We next focus on the power design with $p_{tx}=1$. Recall that under the Rayleigh fading model, we have $p(P_{\mathrm{S}})=e^{-\alpha_{\mathrm{D}}/P_{\mathrm{S}}}$ and $q(P_{\mathrm{S}},\mathbf w_{\mathrm{E}})=e^{-\alpha_{\mathrm{E}}(\mathbf w_{\mathrm{E}})/P_{\mathrm{S}}}$.
In addition, the minimum feasible power is achieved when $p_{tx}=1$, that is,
$P_{\min}(1)=\frac{\alpha_{\mathrm{D}}}{\ln(\Gamma/\lambda)},\ \Gamma>\lambda$. Thus, for a fixed feasible power $P_{\mathrm{S}}$, the secrecy constraint becomes $e^{-\alpha_{\mathrm{E}}(\mathbf w_{\mathrm{E}})/P_{\mathrm{S}}}\le q_{\rm th}(P_{\mathrm{S}},\epsilon)$, where
\begin{equation}
q_{\rm th}(P_{\mathrm{S}},\epsilon)=\frac{(1-\epsilon)p(P_{\mathrm{S}})}{\epsilon+p(P_{\mathrm{S}})-\epsilon p(P_{\mathrm{S}})}.
\label{eq:q-threshold-explicit}
\end{equation} 
Since $0<q_{\rm th}(P_{\mathrm{S}},\epsilon)<1$ for $0<\epsilon<1$, taking the natural logarithm on both sides gives $-\frac{\alpha_{\mathrm{E}}(\mathbf w_{\mathrm{E}})}{P_{\mathrm{S}}}\le \ln q_{\rm th}(P_{\mathrm{S}},\epsilon)$. Substituting $\alpha_{\mathrm{E}}(\mathbf w_{\mathrm{E}})=(2^{R_{\mathrm{tar}}}-1)d_{\mathrm{SE}}^{\eta}\sigma_{\mathrm{E}}^2/A_{\mathrm{E}}$, the SSR boundary can be written as $d_{\mathrm{SE}}\ge r_{\rm SSR}(P_{\mathrm{S}})$, where $r_{\rm SSR}(P_{\mathrm{S}})
=
\left(
\frac{g(P_{\mathrm{S}})A_{\mathrm{E}}}{(2^{R_{\mathrm{tar}}}-1)\sigma_{\mathrm{E}}^2}
\right)^{1/\eta}$, with
\begin{equation}
g(P_{\mathrm{S}})=\alpha_{\mathrm{D}}+P_{\mathrm{S}}\ln\!\left(
\frac{\epsilon}{1-\epsilon}+e^{-\alpha_{\mathrm{D}}/P_{\mathrm{S}}}
\right).
\label{eq:g-ps-correct}
\end{equation}
Therefore, over any bounded observation region, maximizing the SSR is equivalent to minimizing $g(P_{\mathrm{S}})$ over $P_{\mathrm{S}}\in[P_{\min}(1),P_{\max}]$.

\begin{proposition}\label{prop:ps-opt}
Assume $0<\epsilon<1$, $\Gamma>\lambda$, and $P_{\max}\ge P_{\min}(1)$. For $p_{tx}=1$, maximizing the $(\Gamma,\epsilon)$-SSR over a bounded observation region is equivalent to minimizing $g(P_{\mathrm{S}})$ in \eqref{eq:g-ps-correct}~over $P_{\mathrm{S}}\in[P_{\min}(1),P_{\max}]$. The optimal power is
\begin{equation}
P_{\mathrm{S}}^\star=
\begin{cases}
P_{\min}(1), & \epsilon\ge \frac{1}{2}, \\[3pt]
\min\{P_{\max},\max\{P_{\min}(1),P_0\}\}, & 0<\epsilon<\frac{1}{2},
\end{cases}
\end{equation}
where $P_0$ is the unique solution to $g'(P_0)=0$.
\end{proposition}
\begin{proof}
    See Appendix.
\end{proof}

\section{Numerical Results}
In this section, we first validate the analysis of the semantic secrecy probability ${\mathrm{Pr}}(v_{\mathrm{E}}>v_{\mathrm{D}})$ and the average VAoI at the destination $\mathbb{E}[v_{\mathrm{D}}]$. Fig. \ref{fig:va1} shows simulation results for different values of $p$, setting $p_{tx}=0.8$, $\lambda=0.1$ and $q=0.4$. We can find that the simulation results coincide well with the analytical results, validating the correctness of our analysis.

\begin{figure}[htbp]
  \centering
  \includegraphics[width=0.7\columnwidth]{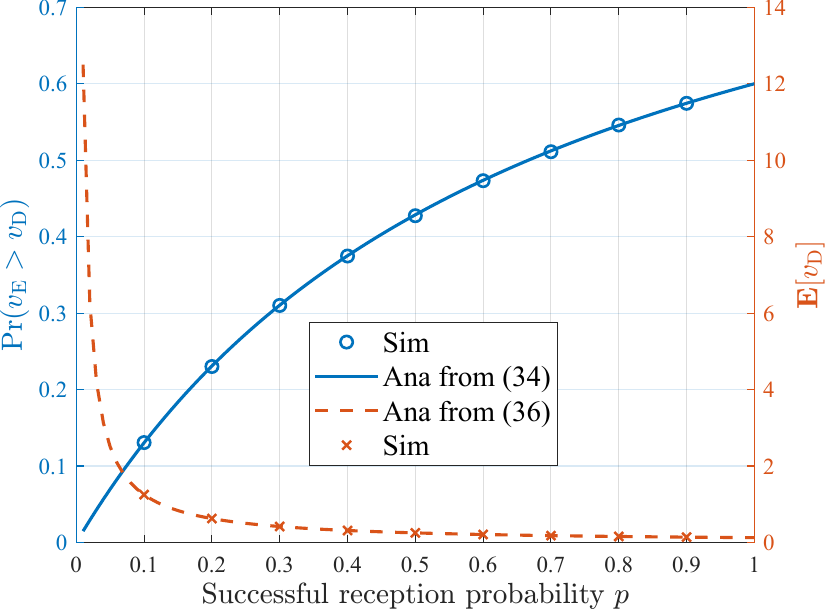}
  \caption{Validation of the analysis for $\mathbb{E}[v_{\mathrm{D}}]$ and ${\mathrm{Pr}}(v_{\mathrm{E}}>v_{\mathrm{D}})$.}
  \label{fig:va1}
\vspace{-0.6cm}
\end{figure}

To evaluate the secrecy performance of the proposed SSR-based design, we consider the following simulation scenario. $\mathrm{S}$ is located at $(0,0)$, while $\mathrm{D}$ is located at $(1,0)$. $\lambda\!=\!0.5$ and $\Gamma\!\!=\!\!2$. The channel bandwidth is set to $10$ MHz. For the fixed-power benchmark, the source transmits with power $P_{\mathrm{S}}\!=\!P_{\max}=1$W. For the optimized scheme, the transmit power $P_{\mathrm{S}}$ is adaptively selected within $[P_{\min}(p_{tx}\!=\!1),P_{\max}]$ according to the optimization developed in Sec III-C. The noise power spectral density is set to $-180$ dBm/Hz, and the large-scale path-loss model (in dB) is given by $PL(d)\!=\!128.1\!+\!37.6\log_{10}(d)$, where $d$ is in kilometers as in \cite{sp2}.

For a given semantic threshold $\Gamma$ at $\mathrm{D}$ and semantic asymmetry requirement $\epsilon$, the SSR is defined as the set of $\mathrm{E}$'s locations for which the destination semantic constraint $\mathbb{E}[v_{\mathrm{D}}]\le \Gamma$ and the semantic asymmetry condition ${\mathrm{Pr}}(v_{\mathrm{E}}>v_{\mathrm{D}})\ge \epsilon$ are simultaneously satisfied. In the simulations, we set $R=20\times 10^6$ bit/s and compare the semantic secrecy regions with and without transmit-power optimization. Fig.~\ref{fig:heatmap-comparison} shows the heatmaps of the maximum supported semantic secrecy requirement $\epsilon$ over the eavesdropper's location for the fixed-power and optimized schemes. In the fixed-power case, the secrecy region is determined by the baseline power $P_{\mathrm{S}}=1$W. In contrast, in the optimized scheme, the transmit power is selected to maximize the SSR under the freshness constraint. It can be observed that the optimized scheme enlarges the secrecy region significantly. Specifically, the contour lines corresponding to the same confidentiality threshold shrink inward after optimization (i.e., the vulnerability area shrinks), indicating that a larger security region can satisfy the preset semantic confidentiality requirement. This result confirms that transmit-power optimization improves the spatial robustness of semantic secrecy.
\begin{figure}[t]
    \centering
    \begin{subfigure}[b]{0.45\columnwidth}
        \centering
        \includegraphics[width=\linewidth]{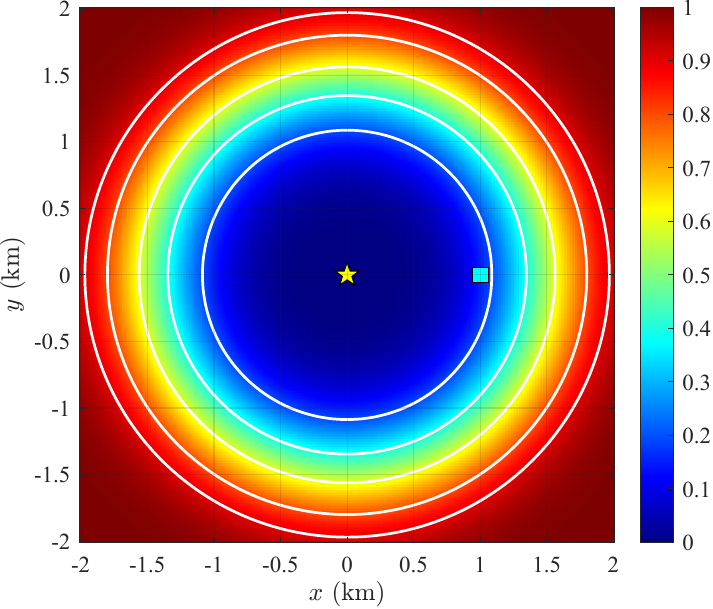}
        \vspace{-0.5cm}
        \caption{Without optimization.}
        \label{fig:heatmap-fixed}
    \end{subfigure}
    \hfill
    \begin{subfigure}[b]{0.45\columnwidth}
        \centering
        \includegraphics[width=\linewidth]{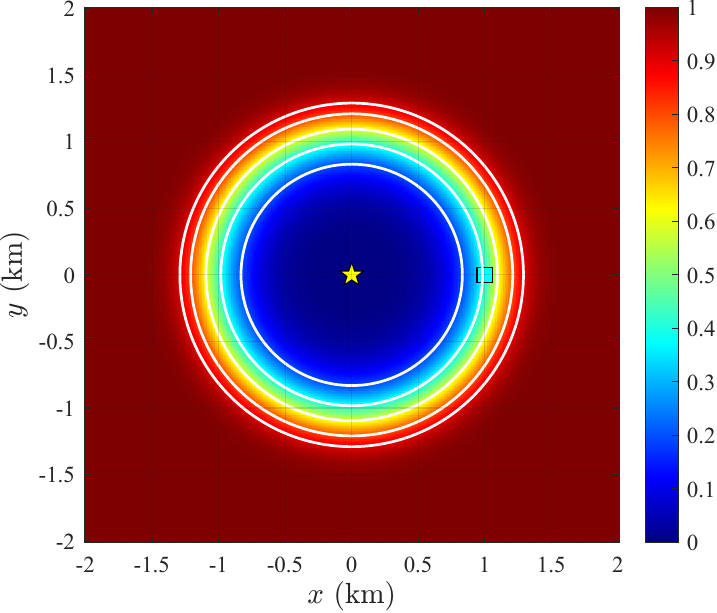}
        \vspace{-0.5cm}
        \caption{With optimization.}
        \label{fig:heatmap-optimized}
    \end{subfigure}
    \caption{Heatmaps of the maximum achievable semantic secrecy requirement $\epsilon$ across possible locations of $\mathrm{E}$, with and without transmit‑power optimization. The white circles correspond to $\epsilon\in[0.2,0.4,0.6,0.8,0.9]$ from the center outward. The star and square represent the location of $\mathrm{S}$ and $\mathrm{D}$, respectively.}
    \label{fig:heatmap-comparison}
    \vspace{-0.6cm}
\end{figure}
\section{Conclusions}
This paper examined semantic secrecy in wireless status‑updating systems from the perspective of the version age of information (VAoI). Using a two‑dimensional Markov chain, we derived closed‑form results for the destination’s average VAoI and the probability of semantic asymmetry, enabling the definition of the $(\Gamma,\epsilon)$-semantic secrecy region (SSR). The analysis showed that the transmission probability constrains the feasible power through freshness requirements, while transmit power shapes the secrecy–semantic tradeoff and the SSR boundary. Numerical results confirmed that optimized power allocation significantly enlarges the secrecy region compared with fixed‑power schemes.
\vspace{-0.15cm}

\bibliography{ref}

@INPROCEEDINGS{9517796,
  author={Yates, Roy D.},
  booktitle={IEEE ISIT}, 
  title={The Age of Gossip in Networks}, 
  year={2021},
  volume={},
  number={},
  pages={},
  doi={10.1109/ISIT45174.2021.9517796}}

@INPROCEEDINGS{10620867,
  author={Wang, Qian and Chen, He and Mohapatra, Parthajit and Pappas, Nikolaos},
  booktitle={IEEE INFOCOM Workshops}, 
  title={Secure Status Updates under Eavesdropping: Age of Information-based Secrecy Metrics}, 
  year={2024},
  volume={},
  number={},
  pages={},
  doi={10.1109/INFOCOMWKSHPS61880.2024.10620867}}

@article{sp3,
  title={Secure transmission in mmwave wiretap channels: On sector guard zone and blockages},
  author={Song, Yi and Yang, Weiwei and Xiang, Zhongwu and Liu, Yiliang and Cai, Yueming},
  journal={Entropy},
  volume={21},
  number={4},
  pages={427},
  year={2019},
  publisher={MDPI}
}

@ARTICLE{sp2,
  author={Li, Wei and Ghogho, Mounir and Chen, Bin and Xiong, Chunlin},
  journal={IEEE Communications Letters}, 
  title={Secure Communication via Sending Artificial Noise by the Receiver: Outage Secrecy Capacity/Region Analysis}, 
  year={2012},
  volume={16},
  number={10},
  pages={},
  doi={10.1109/LCOMM.2012.081612.121344}}

@INPROCEEDINGS{sp1,
  author={Marina, Ninoslav and Hjorungnes, Are},
  booktitle={IEEE WCNC}, 
  title={Characterization of the Secrecy Region of a Single Relay Cooperative System}, 
  year={2010},
  volume={},
  number={},
  pages={},
  doi={10.1109/WCNC.2010.5506566}}

@ARTICLE{aoiec2,
  author={Wang, Chao and Li, Zan and Zheng, Tong-Xing and Ng, Derrick Wing Kwan and Al-Dhahir, Naofal},
  journal={IEEE Transactions on Wireless Communications}, 
  title={Intelligent Reflecting Surface-Aided Full-Duplex Covert Communications: Information Freshness Optimization}, 
  year={2023},
  volume={22},
  number={5},
  pages={},
  doi={10.1109/TWC.2022.3217041}}

@ARTICLE{aoiec1,
  author={Wang, Yida and Yan, Shihao and Yang, Weiwei and Cai, Yueming},
  journal={IEEE Wireless Communications Letters}, 
  title={Covert Communications With Constrained Age of Information}, 
  year={2021},
  volume={10},
  number={2},
  pages={},
  doi={10.1109/LWC.2020.3031492}}

@ARTICLE{aoie2,
  author={Zheng, Lei and Ren, Juanjuan and Liu, Yong and Chen, Qingchun},
  journal={IEEE Wireless Communications Letters}, 
  title={Analysis and Optimization of Age of Information and Age of Leaked Information for {IoT} Networks}, 
  year={2024},
  volume={13},
  number={12},
  pages={},
  doi={10.1109/LWC.2024.3481059}}

@INPROCEEDINGS{aoie1,
  author={Crosara, Laura and Laurenti, Nicola and Badia, Leonardo},
  booktitle={BalkanCom}, 
  title={It Is Rude to Ask a Sensor Its Age-of-Information: Status Updates Against an Eavesdropping Node}, 
  year={2023},
  volume={},
  number={},
  pages={},
  doi={10.1109/BalkanCom58402.2023.10167914}}

@INPROCEEDINGS{Kaul2012real,
  author={Kaul, Sanjit and Yates, Roy and Gruteser, Marco},
  booktitle={IEEE INFOCOM}, 
  title={Real-time status: How often should one update?}, 
  year={2012},
  volume={},
  number={},
  pages={},
  doi={10.1109/INFCOM.2012.6195689}}

@book{gradshteyn2014table,
	title={Table of integrals, series, and products},
	author={Gradshteyn, Izrail Solomonovich and Ryzhik, Iosif Moiseevich},
	year={2014},
	publisher={Academic press}
}

@article{sicari2015security,
	title={Security, privacy and trust in {I}nternet of {T}hings: The road ahead},
	author={Sicari, Sabrina and Rizzardi, Alessandra and Grieco, Luigi Alfredo and Coen-Porisini, Alberto},
	journal={Computer networks},
	volume={76},
	pages={},
	year={2015},
	publisher={Elsevier}
}

@article{wyner1975wire,
	title={The wire-tap channel},
	author={Wyner, Aaron D},
	journal={Bell system technical journal},
	volume={54},
	number={8},
	pages={},
	year={1975},
	publisher={Wiley Online Library}
}

@book{bloch2011physical,
	title={Physical-layer security: from information theory to security engineering},
	author={Bloch, Matthieu and Barros, Joao},
	year={2011},
	publisher={Cambridge University Press}
}

@ARTICLE{ageversion,
  author={Salimnejad, Mehrdad and Kountouris, Marios and Ephremides, Anthony and Pappas, Nikolaos},
  journal={IEEE Transactions on Communications}, 
  title={{A}ge of {I}nformation {V}ersions: A Semantic View of {M}arkov Source Monitoring}, 
  year={2025},
  volume={73},
  number={12},
  pages={},
  doi={10.1109/TCOMM.2025.3616209}}

@book{mohapatra2024physical,
  title={Physical-layer Security for {6G}},
  author={Mohapatra, Parthajit and Pappas, Nikolaos and Chorti, Arsenia and Tomasin, Stefano},
  year={2024},
  publisher={John Wiley \& Sons}
}
\bibliographystyle{IEEEtran}

\appendix

\section{Proof of Proposition 2}
Note that $r_{\rm SSR}(P_{\mathrm{S}})$ is monotone in $g(P_{\mathrm{S}})$, maximizing the SSR is equivalent to minimizing $g(P_{\mathrm{S}})$. Let $c=\epsilon/(1-\epsilon)$ and $u=e^{-\alpha_{\mathrm{D}}/P_{\mathrm{S}}}$. Then
\begin{equation}
g'(P_{\mathrm{S}})=\ln(c+u)-\frac{u\ln u}{c+u}.
\end{equation}
Moreover,
\begin{equation}
\frac{d g'}{d u}=-\frac{c\ln u}{(c+u)^2}>0,
\qquad
\frac{du}{dP_{\mathrm{S}}}=\frac{\alpha_{\mathrm{D}}}{P_{\mathrm{S}}^2}u>0,
\end{equation}
which shows that $g'(P_{\mathrm{S}})$ is strictly increasing in $P_{\mathrm{S}}$. Also,
\begin{equation}
\lim_{P_{\mathrm{S}} \rightarrow 0^{+} }g'(P_{\mathrm{S}})=\ln\frac{\epsilon}{1-\epsilon},
\qquad
\lim_{P_{\mathrm{S}}\rightarrow\infty }g'(P_{\mathrm{S}})=\ln\frac{1}{1-\epsilon}>0.
\end{equation}
Hence, if $\epsilon\ge 1/2$, then $g'(P_{\mathrm{S}})\ge 0$ and the minimum is attained at $P_{\min}(1)$. If $0<\epsilon<1/2$, then $g'(P_{\mathrm{S}})$ crosses zero exactly once, yielding the unique unconstrained minimizer $P_0$; the constrained optimum is its projection onto $[P_{\min}(1),P_{\max}]$. This completes the proof.


\end{document}